\documentclass[runningheads,a4paper]{llncs}

\usepackage{amssymb, amsmath, mathtools}
\usepackage{graphicx, comment, epstopdf}
\usepackage{algorithm, algorithmicx, algcompatible, algpseudocode}
\usepackage[table]{xcolor}
\definecolor{Gray}{gray}{0.90}
\usepackage{color}
\usepackage{tikz, pgfplots, epstopdf, todonotes}

\DeclareMathOperator*{\argmin}{arg\,min}

\newcommand{\veta}[0]{\vect{\eta}}

\newcommand{\vect}[1]{\ensuremath{\boldsymbol{#1}}}
\newcommand{\calD}[0]{\mathcal{D}}
\newcommand{\calL}[0]{\mathcal{L}}
\newcommand{\calP}[0]{\mathcal{P}}
\newcommand{\calO}[0]{\mathcal{O}}
\newcommand{\FP}[0]{\ensuremath{F\!P}}   
\newcommand{\Real}[0]{\ensuremath{\mathbb{R}}}
\newcommand{\keywords}[1]{\par\addvspace\baselineskip
\noindent\keywordname\enspace\ignorespaces#1}
\newcommand{\algoname}[1]{\textnormal{\textsc{#1}}}

\usepackage{url}
\urldef{\mailsa}\path|{firstname.lastname}@vanderbilt.edu|
\urldef{\mailsb}\path|laszka@berkeley.edu| 

\begin{document}

\mainmatter

\title{Optimal Thresholds for Anomaly-Based Intrusion Detection in Dynamical Environments}
\titlerunning{Optimal Thresholds for Intrusion Detection in Dynamical Environments}

\author{Amin Ghafouri\inst{1}
	\and Waseem Abbas\inst{1} \and  Aron Laszka\inst{2} \and Yevgeniy Vorobeychik\inst{1} \and Xenofon Koutsoukos\inst{1}}
\institute{Institute for Software Integrated Systems,\\
	Vanderbilt University, USA\\
	\mailsa\\
	\and
	Department of Electrical Engineering and Computer Sciences, \\ University of California, Berkeley, USA\\
	\mailsb\\
}
\authorrunning{Ghafouri et al.}
\maketitle

\begin{abstract} 
In cyber-physical systems, malicious and resourceful attackers could penetrate a system through cyber means and cause significant physical damage. Consequently, early detection of such attacks becomes integral towards making these systems resilient to attacks. To achieve this objective, intrusion detection systems (IDS) that are able to detect malicious behavior early enough can be deployed. However, practical IDS are imperfect and sometimes they may produce false alarms even for normal system behavior. Since alarms need to be investigated for any potential damage, a large number of false alarms may increase the operational costs significantly. Thus, IDS need to be configured properly, as oversensitive IDS could detect attacks very early but at the cost of a higher number of false alarms. Similarly, IDS with very low sensitivity could reduce the false alarms while increasing the time to detect the attacks. The configuration of IDS to strike the right balance between time to detecting attacks and the rate of false positives is a challenging task, especially in dynamic environments, in which the damage caused by a successful attack is time-varying.

In this paper, using a game-theoretic setup, we study the problem of finding optimal detection thresholds for anomaly-based detectors implemented in dynamical systems in the face of strategic attacks. We formulate the problem as an attacker-defender security game, and determine thresholds for the detector to achieve an optimal trade-off between the detection delay and the false positive rates. In this direction, we first provide an algorithm that computes an optimal \textit{fixed threshold} that remains fixed throughout. Second, we allow the detector's threshold to change with time to further minimize the defender's loss, and we provide a polynomial-time algorithm to compute time-varying thresholds, which we call \textit{adaptive thresholds}. Finally, we numerically evaluate our results using a water-distribution network as a case study.

\end{abstract}

\section{Introduction}

In recent years, we have seen an increasing trend of malicious intruders and attackers penetrating into various cyber-physical systems (CPS) through cyber means and causing severe physical damage. Examples of such incidents include the infamous Stuxnet worm \cite{kushner2013real}, cyber attack on German steel plant \cite{lee2014german}, and Maroochy Shire water-services incident \cite{abrams2008malicious} to name a few. To maximize the damage, attackers often aim to remain covert and avoid getting detected for an extended duration of time. As a result, it becomes crucial for a defender to design and place efficient intrusion and attack detection mechanisms to minimize the damage. While attackers may be able to hide the specific information technology methods used to exploit and reprogram a CPS, they cannot hide their final intent: the need to cause an adverse effect on the CPS by sending malicious sensor or controller data that will not match the behavior expected by an anomaly-based detection system \cite{cardenas2011attacks}. Anomaly-based detection systems incorporate knowledge of the physical system, in order to monitor the system for suspicious activities and cyber-attacks.
An important design consideration in such detection systems is to carefully configure them in order to satisfy the expected monitoring goals.

A well-known method for anomaly-based detection is sequential change detection \cite{kailath1998detection}. This method assumes a sequence of measurements that starts under the normal hypothesis and then, at some point in time, it changes to the attack hypothesis. Change detection algorithm attempts to detect this change as soon as possible. In a sequential change detection, there is a \textit{detection delay}, that is, a time difference between when an attack occurs and when an alarm is raised. On the other hand, detection algorithms may induce \textit{false positives}, that is, alarms raised for normal system behavior. In general, it is desirable to reduce detection delay as much as possible while maintaining an acceptable false positive rate. Nevertheless, there exists a trade-off between the detection delay and the rate of false positives, which can be controlled by changing the sensitivity of the the detector. A typical way to control detector sensitivity is through a detection threshold: by decreasing (increasing) detection threshold, a defender can decrease (increase) detection delay and increase (decrease) false positive rate. Consequently, the detection threshold must be carefully selected, since a large value may result in excessive losses due to high detection delays, while a small value may result in wasting operational resources on investigating false alarms. 

Finding an \textit{optimal threshold}, that is, one that optimally balances the detection delay-false positive trade-off, is a challenging problem \cite{laszka2016optimal}. However, it becomes much more challenging when detectors are deployed in CPS with dynamic behavior, that is, when the expected damage incurred from undetected cyber-attacks depends on the system state and time. 
As a result, an attack on a CPS which is in a critical state is expected to cause more damage as compared to an attack in a less critical state. For example, in water distribution networks and electrical grids, disruptions at a high-demand time are more problematic than disruptions at a low-demand time. Hence, defenders need to incorporate time-dependent information in computing optimal detection thresholds when facing strategic attackers. 

We study the problem of finding optimal detection thresholds for anomaly-based detectors implemented in dynamical systems in the face of strategic attacks. We model rational attacks against a system that is equipped with a detector as a two-player game between a defender and an attacker. We assume that an attacker can attack a system at any time. Considering that the damage is time-dependent, the attacker's objective is to choose the optimal time to launch an attack to maximize the damage incurred. On the other hand, the defender's objective is to select the detection thresholds to detect an attack with minimum delay while maintaining an acceptable rate of false positives. To this end, first we present an algorithm that selects an optimal threshold for the detector that is independent of time (i.e., \emph{fixed}). We call it as a \textit{fixed threshold} strategy. Next, we allow the defender to select a time-varying threshold while associating a cost with the threshold change. For this purpose, we present a polynomial time algorithm that computes thresholds that may depend on time. We call this approach the \textit{adaptive threshold} strategy. We present a detailed analysis of the computational complexity and performance of both the fixed and adaptive threshold strategies. Finally, we evaluate our results using a water distribution system as a case study. Since expected damage to the system by an attack is time-dependent, the adaptive threshold strategy achieves a better overall detection delay-false positive trade-off, and consequently minimize the defender's losses. Our simulations indicate that this is indeed the case, and adaptive thresholds outperform the fixed threshold.

The remainder of this paper is organized as follows. In Section~2, we introduce our system model. In Section~3, we present our game-theoretic model and define optimal fixed and adaptive detection thresholds. In Section~4, we analyze both strategies and present algorithms to obtain optimal fixed and adaptive thresholds.
In Section~5, we evaluate these algorithms using numerical example. In Section~6, we discuss related work on detection threshold selection in the face of strategic attacks. Finally, we offer concluding remarks in Section~7.

\section{System Model} 
In this section, we present the system model. For a list of symbols used in this paper, see Table~\ref{tab:symbols}.

\begin{table}[h]
	\centering
	\caption{List of Symbols}
	\label{tab:symbols}
	\renewcommand*{\arraystretch}{1.2}
	\begin{tabular}{| c | p{10cm} |}
		\hline Symbol & Description \\ \hline \hline
		\rowcolor{Gray}$\calD(k)$ & expected damage caused by an attack at timestep $k$ \\
		$\delta(\eta)$ & expected detection delay given detection threshold $\eta$ \\
		\rowcolor{Gray}$\FP(\eta)$ & false positive rate given detection threshold is $\eta$ \\
		$C$ & cost of false alarms\\
		\rowcolor{Gray} $\calP(\eta, k_a)$ & attacker's payoff for threshold $\eta$ and attack time $k_a$ \\
		$\calL(\eta, k_a)$ & defender's loss for threshold $\eta$ and attack time $k_a$ \\  \hline
		\multicolumn{2}{|c|}{Adaptive Threshold} \\ \hline
		\rowcolor{Gray} $\calP(\veta, k_a)$ & attacker's payoff for adaptive threshold $\veta=\{\eta_k\}$ and attack time $k_a$ \\
		$\calL(\veta, k_a)$ & defender's loss for adaptive threshold $\veta=\{\eta_k\}$ and attack time $k_a$ \\
		\hline
	\end{tabular}
\end{table}

\subsection{Attack Model}
Let the system have a finite discrete time horizon of interest denoted by $\{1,...,T\}$. Adversaries may exploit threat channels by compromising the system through a deception attack that starts at time $k_a$ and ends at $k_e$, thus spanning over the interval $[k_a, k_e]$.
Deception attacks are the ones that result in loss of integrity of sensor-control data, and their corresponding danger is especially profound due to the tight coupling of physical and cyber components (see \cite{amin2013quest} for details). 
If an attack remains undetected, it will enable the attacker to cause physical or financial damage. In order to represent the tight relation between the CPS's dynamic behavior and the expected loss incurred from undetected attacks, we model the potential damage of an attack as a function of time.

\begin{definition}
	(Expected Damage Function): Damage function of a CPS is a function $\calD:~\{1,...,T\} \rightarrow \Real_+$, which represents the expected damage $\calD(k)$ incurred to a system from an undetected attack at time $k \in \{1,...,T\}$.
\end{definition}

The definition above describes instant damage at a time $k \in \{1,...,T\}$. Following this definition, expected total damage resulting from an attack that spans over some interval is defined as follows.

\begin{definition}
	(Expected Total Damage): Expected total damage is denoted by a function $\bar{\calD}:\{1,...,T\} \times \{1,...,T\} \rightarrow \Real_+$, which represents the expected total damage $\bar{\calD}(k_a,k_e)$ incurred to a system from an undetected attack in a period $[k_a,k_e]$. Formally,
\begin{equation}
\bar{\calD}(k_a,k_e) = \sum_{k=k_a}^{k_e} \calD(k) \;.
\end{equation}	
\end{definition}

\subsection{Detector}

We consider a defender whose objective is to protect the physical system, which is equipped with a detector. The detector's goal is to determine whether a sequence of received measurements generated through the system corresponds to the normal behavior or an attack. To implement a detection algorithm, we utilize a widely used method known as sequential change detection \cite{kailath1998detection}. This method assumes a sequence of measurements that starts under the normal hypothesis, and then, at some point in time, changes to the attack hypothesis. Change detection algorithm attempts to detect this change as soon as possible. 

\subsubsection{Example (CUSUM).} The Cumulative sum (CUSUM) is a statistic used for change detection. The nonparametric CUSUM statistic $S(k)$ is described by
\begin{equation*}
S(k)= (S(k-1) + z(k))^+,
\end{equation*}
where $S(0)=0$, $(a)^+ = a $ if $a \geq 0$ and zero otherwise, and $z(k)$ is generated by an observer, such that under normal behavior it has expected value of less than zero \cite{cardenas2011attacks}. Assigning $\eta$ as the detection threshold chosen based on a desired false alarm rate, the corresponding decision rule is defined as
\begin{equation*}
d(S(k)) = \left\{ \begin{array}{lcl}
\textrm{Attack} & \textrm{  if } S(k)>\eta \\ \textrm{Normal} & \textrm{otherwise} \\
\end{array}\right.
\end{equation*}

\subsubsection{Detection Delay and False Positive Rate.} 
In detectors implementing change detection, there might be a \textit{detection delay}, which is the time taken by the detector to raise an alarm since the occurrence of an attack.\footnote{Note that any desired definition of detection delay may be considered, for example, stationary average delay \cite{shiryaev1961problem,srivastava1993comparison}.} Further, there might be a \textit{false positive}, which means raising an alarm when the system exhibits normal behavior.
In general, it is desirable to reduce detection delay as much as possible while maintaining an acceptable false positive rate. But, there exists a trade-off between the detection delay and the rate of false positives, which can be controlled by changing the detection threshold. In particular, by decreasing (increasing) the detection threshold, a defender can decrease (increase) detection delay and increase (decrease) false positive rate. Finding the optimal trade-off point and its corresponding \textit{optimal threshold} is known to be an important problem \cite{laszka2016optimal}, however, it is much more important in CPS, since expected damage incurred from undetected attack directly depends on detector's performance. 

We represent detection delay by the continuous function $\delta:~\Real_+ \to~\mathbb{N} \cup \{0\}$, where $\delta(\eta)$ is the detection delay (in timesteps) when threshold is $\eta$. Further, we denote the attainable false positive rate by the continuous function $\FP: \Real_+ \to [0,1]$, where $\FP(\eta)$ is the false positive rate when the detection threshold is $\eta$. We assume that $\delta$ is increasing and $\FP$ is decreasing, which is true for most typical detectors including the CUSUM detector.

\section{Problem Statement}
In this section, we present the optimal threshold selection problem. We consider two cases: 1) Fixed threshold, in which the defender selects an optimal threshold and then keeps it fixed; and 2) Adaptive threshold, in which the defender changes detection threshold based on time. We model this problems as conflicts between a defender and an attacker, which are formulated as two-player Stackelberg security games.

\subsection{Fixed Threshold}

\subsubsection{Strategic Choices.}
The defender's strategic choice is to select a detection threshold $\eta$. The resulting detection delay and false positive rate are $\delta(\eta)$ and $\FP(\eta)$, respectively. We consider the worst-case attacker that will not stop the attack before detection in order to maximize the damage. Consequently, the attacker's strategic choice becomes to select a time $k_a$ to start the attack. Note that we consider damage from only undetected attacks since the mitigation of non-stealthy attacks is independent of detector.

\subsubsection{Defender's Loss and Attacker's Payoff.}
As an alarm is raised, the defender needs to investigate the system to determine whether an attack has actually happened, which will cost him $C$. When the defender selects threshold $\eta$ and the attacker starts its attack at a timestep $k_a$, the defender's loss (i.e., inverse payoff)~is
\begin{equation}
\label{eq:defenderLoss}
\calL(\eta, k_a) = C \cdot \FP (\eta) \cdot T + \sum_{k=k_a}^{k_a + \delta(\eta)} \calD(k) \;,
\end{equation} 
that is, the amount of resources wasted on manually investigating false positives and the expected amount of damage caused by undetected attacks.

For the strategies $(\eta, k_a)$, the attacker's payoff is
\begin{equation}
\label{eq:attackerUtility}
\calP(\eta, k_a) = \sum_{k=k_a}^{k_a + \delta(\eta)} \calD(k) \;.
\end{equation}
that is, the total damage incurred to the system prior to the expected detection time. The idea behind this payoff function is the assumption of a worst-case attacker that has the goal of maximizing the damage.

\subsubsection{Best-Response Attack and Optimal Fixed Threshold.}
We assume that the attacker knows the system model and defender's strategy, and can thus compute the detection threshold chosen by the defender. Hence, the attacker will play a \emph{best-response} attack to the defender's strategy, which is defined below.

\begin{definition}
	(Best-Response Attack): Taking the defender's strategy as given, the attacker's strategy is a best-response if it maximizes the attacker's payoff.
	Formally, an attack starting at $k_a$ is a best-response attack given a defense strategy $\eta$, if it maximizes $\calP(\eta, k_a)$.
\end{definition}

Further, the defender must choose his strategy expecting that the attacker will play a best-response. We formulate the defender's optimal strategy as strong Stackelberg equilibrium (SSE)~\cite{korzhyk2011stackelberg}, which is commonly used in the security literature for solving Stackelberg games.

\begin{definition}
	\label{def:optimalDefense}
	(Optimal Fixed Threshold): We call a defense strategy optimal if it minimizes the defender's loss given that the attacker always plays a best-response.  
	Formally, an optimal defense is
	\begin{equation}
	\argmin_{\substack{\eta, \\ k_a \in \text{bestResponses}(\eta)}} \calL(\eta, k_a) ,
	\end{equation}
	where $\text{bestResponses}(\eta)$ is the set of best-response attacks against $\eta$.
\end{definition}

\subsection{Adaptive Threshold} 

Although the optimal fixed threshold minimizes the defender's loss considering attacks at critical periods (i.e., periods with high damage), it imposes a high false alarm rate at less critical periods. Adaptive threshold strategies directly address this issue. The idea of adaptive threshold is to reduce the detector's sensitivity during less critical periods (via increasing the threshold), and increase the sensitivity during more critical periods (via decreasing the threshold). As it will be shown, this significantly decreases the loss corresponding to false alarms. 
However, the defender may not want to continuously change the threshold, since a threshold change requires a reconfiguration of the detector that has a cost. Hence, the rational defender needs to find an \textit{optimal adaptive threshold}, which is a balance between continuously changing the threshold and keeping it fixed.

The adaptive threshold is defined by $\veta = \{\eta_k\}_{k=1}^{T}$. The number of threshold changes is described by $N=|S|$, where $S = \{ k \, | \, \eta_k \neq \eta_{k+1}, k\in \{1,...,T-1\} \}$. If the system is under an undetected attack, the detection delay for each timestep $k$ is the delay corresponding to its threshold, i.e., $\delta(\eta_k)$. We define detection time of an attack $k_a$ as the time index at which the attack is first detected. It is given by
\begin{equation}
\sigma(\veta, k_a)=\{\min k \, | \, \delta(\eta_k) \leq k - k_a \} \;.
\end{equation}
Note that the equation above represents the time index at which the attack is first detected, and not the detection delay. The detection delay for an attack $k_a$ can be obtained by $\delta(\veta, k_a)\coloneqq \sigma(\veta, k_a)-k_a$.

\subsubsection{Strategic Choices.}
The defender's strategic choice is to select the threshold for each time index, given by $\veta=\{\eta_1,\eta_2,...,\eta_T\}$. We call ${\veta}$ to be the set of adaptive threshold. Since we consider a worst-case attacker that will not stop the attack before detection, the attacker's strategic choice is to select a time $k_a$ to start the attack. 

\subsubsection{Defender's Loss and Attacker's Payoff}
Let $C_d$ be the cost associated with each threshold change. When the defender selects adaptive threshold $\veta$, and the attacker starts its attack at a timestep $k_a$, the defender's loss is 
\begin{equation}
	\label{eq:adDefenderLoss}
	\calL(\veta, k_a) = N \cdot C_d + \sum_{k=1}^{T} C \cdot \FP (\eta_k) + \sum_{k=k_a}^{\sigma(\veta, k_a)} \calD(k) \;,
\end{equation} 
that is, the amount of resources spent on changing the threshold, operational costs of manually investigating false alarms, and the expected amount of damage caused by undetected attacks.

For the strategies $(\veta, k_a)$, the attacker's payoff is the total damage prior to the expected detection time,
\begin{equation}
	\label{eq:adAttackerUtility}
	\calP(\veta, k_a)=\sum_{k=k_a}^{\sigma(\veta, k_a)}  \calD(k) \;.
\end{equation}

\subsubsection{Best-Response Attack and Optimal Adaptive Threshold.}
The definitions presented in this part are analogous to the ones discussed above for the case of optimal fixed threshold. We assume the attacker can compute the adaptive threshold, and will play a \emph{best-response} to the defender's strategy, as defined below.

\begin{definition}
	(Best-Response Attack): Taking the defender's strategy as given, the attacker's strategy is a best-response if it maximizes the attacker's payoff.
	Formally, an attack $k_a$ is a best-response given a defense strategy $\veta$, if it maximizes $\calP(\veta, k_a)$ as defined in \eqref{eq:adAttackerUtility}.
\end{definition}

Further, the defender must choose its strategy expecting that the attacker will play a best-response.

\begin{definition}
	\label{def:adOptimalDefense}
	(Optimal Adaptive Threshold): We call a defense strategy \emph{optimal} if it minimizes the defender's loss given that the attacker always plays a best-response with tie-breaking in favor of the defender.  
	Formally, an optimal defense is
	\begin{equation}
		\argmin_{\substack{\veta, \\ k_a \in \text{bestResponses}(\veta)}} \calL(\veta, k_a) ,
	\end{equation}
	where $\text{bestResponses}(\veta)$ is the best-response attack against $\veta$.
\end{definition}

\section{Selection of Optimal Thresholds}
\label{sec:analysis}

In this section, we present polynomial-time algorithms to compute optimal thresholds, both for the fixed and adaptive cases.

\subsection{Fixed Threshold}
\label{sec:fixedThreshold}

To compute an optimal fixed threshold, we present Algorithm~\ref{algo:fixed}. Here, we consider that any detection delay can be achieved by selecting a specific threshold value. Therefore, the algorithm finds an optimal detection delay, from which the optimal threshold value can be selected.
To find the optimal detection delay, the algorithm iterates through all possible values of detection delay and selects the one that minimizes the defender's loss considering a best-response attack. To find a best-response attack $k_a$, given a delay $\delta$, the algorithm iterates through all possible values of $k_a$, and selects the one that maximizes the payoff.

\begin{algorithm}[h!]
	\caption{Algorithm for Optimal Fixed Threshold}
	\label{algo:fixed}
	\begin{algorithmic}[1]
		\State \textbf{Input} $\calD(k)$, $T$, $C$
		\State \textbf{Initialize:} $\delta \gets 0$, $L^\ast \gets \infty $
		\While {$\delta < T$}
		\State $k_a \gets 1$, $P' \gets 0$
		\While {$k_a < T$}  \label{algo1:ka_it}
		\State $P(\delta,k_a) \gets \sum_{k_a}^{k_a+\delta} D(k)$ \label{algo1:payoff}
		\If {$P(\delta,k_a) > P'$} \label{algo1:optpayoff1}
		\State $P' \gets P(\delta,k_a)$
		\State $L' \gets P' + C \cdot \FP(\delta) \cdot T$ \label{algo1:loss}
		\EndIf \label{algo1:optpayoff2}
		\State $k_a \gets k_a+1$
		\EndWhile
		\If {$L' < L^\ast$} \label{algo1:optloss1}
		\State $L^\ast \gets L'$
		\State $\delta^\ast \gets \delta$
		\EndIf \label{algo1:optloss2}
		\State $\delta \gets \delta + 1$  \label{algo1:l+1}
		\EndWhile
		\State \textbf{return} $\delta^\ast$
	\end{algorithmic}
\end{algorithm}

\begin{proposition}
	Algorithm~\ref{algo:fixed} computes an optimal fixed threshold in $\calO(T^2)$ steps.
\end{proposition}	
\begin{proof}
	The obtained threshold is optimal since the algorithm evaluates all possible solutions through exhaustive search. Given a pair $(\delta,k_a)$, when computing the attacker's payoff $P(\delta,k_a)$ in Line~\ref{algo1:payoff}, we use the payoff computed in previous iteration, and write $P(\delta,k_a) = P(\delta,k_a-1)+ \calD(k_a-1) + \calD(k_a+\delta)$, which takes constant time. Therefore, the running time of the algorithm is subquadratic in the total number of timesteps $T$. \qed
\end{proof}

\subsection{Adaptive Threshold}

We present Algorithm~\ref{algo:adaptive} for finding optimal adaptive thresholds for any instance of the attacker-defender game, which is based on the SSE. The approach comprises 1) a dynamic-programming algorithm for finding minimum-cost thresholds subject to the constraint that the damage caused by a worst-case attack is at most a given value and 2) an exhaustive search, which finds an optimal damage value and thereby optimal thresholds. 
For ease of presentation, we use detection delays $\delta_k$ and the corresponding maximal thresholds $\eta_k$ interchangeably (e.g., we let $\FP(\delta_k)$ denote the false-positive rate of the maximal threshold that results in detection delay $\delta_k$), and we let $\Delta$ denote the set of all attainable detection delay values.

\begin{theorem}
	Algorithm~\ref{algo:adaptive} computes an optimal adaptive threshold.
\end{theorem}

\begin{algorithm}[h!]
	\caption{Algorithm for Optimal Adaptive Thresholds}
	\label{algo:adaptive}
	\begin{algorithmic}[1]
		\State \textbf{Input} $\calD(k)$, $T$, $C$
		\State $\text{SearchSpace} \gets \left\{ \sum_{k = k_a}^{k_e} D(k) ~\middle|~ k_a \in \{1, \ldots, T-1\}, ~k_e \in \{n+1, \ldots, T\}\right\}$
		\ForAll {$P \in \text{SearchSpace}$}
		  \State $TC(P), \delta^*_1(P), \ldots, \delta^*_T(P) \gets \algoname{MinimumCostThresholds}(P)$
		\EndFor
		\State {$P^* \gets \argmin_{P \in \text{SearchSpace}} TC(P)$}
		\State $\textbf{return } \delta^*_1(P^*), \ldots, \delta^*_T(P^*)$
		\Statex
	  \Function {MinimumCostThresholds}{$P$}
	    \State $\forall~ m \in \{0, \ldots, T-1\},~ \delta \in \Delta:~ \algoname{Cost}(T+1, m, \delta) \gets 0$
	    \For {$n = T, \ldots, 1$}
	      \ForAll {$m \in \{0, \ldots n - 1\}$}
	        \ForAll {$\delta_{n-1} \in \Delta$}
	          \ForAll {$\delta_n \in \Delta$}
	            \If {$\delta_n > m$}
	              \State $S(\delta_n) \gets \algoname{Cost}(n+1, m+1, \delta_n) + C \cdot \FP(\delta_n)$
	            \ElsIf {$\sum_{k=n-m}^n \calD(k) \leq P$}
	              \State $S(\delta_n) \gets \algoname{Cost}_P(n+1, \delta_n, \delta_n) + C \cdot \FP(\delta_n)$
	            \Else
	              \State $S(\delta_n) \gets \infty$
              \EndIf
	            \If {$\delta_{n-1} \neq \delta_n \land n > 1$}
	              \State $S(\delta_n) \gets S(\delta_n) + C_d$
              \EndIf
	          \EndFor
	          \State $\delta^*(n, m, \delta_{n-1}) \gets \argmin_{\delta_n} S(\delta_n)$
	          \State $\algoname{Cost}(n, m, \delta_{n -1}) \gets \min_{\delta_n} S(\delta_n)$
	        \EndFor
	      \EndFor
	    \EndFor
	    \State $m \gets 0,~ \delta^*_0 \gets \text{arbitrary}$
	    \ForAll {$n = 1, \ldots T$}
	      \State $\delta^*_n \gets \delta^*(n, m, \delta^*_{n-1})$
	      \State $m \gets \min\{m + 1, \delta^*_n\}$
	    \EndFor
	    \State $\textbf{return } \algoname{Cost}(1, 0, \text{arbitrary}), \delta^*_1, \ldots, \delta^*_T$
	  \EndFunction	
	\end{algorithmic}
\end{algorithm}

\begin{proof}[Sketch.]
First, we prove that our dynamic-programming algorithm, called $\algoname{MinimumCostThresholds}$ in Algorithm~\ref{algo:adaptive}, finds minimum-cost thresholds subject to any damage constraint $P$. Then, we show that our exhaustive search finds an optimal damage constraint $P$, which minimizes the defender's loss given that the attacker plays a best response.

\subsubsection{1) Minimum-Cost Thresholds}
In the first part, we assume that we are given a damage value $P$, and we have to find thresholds that minimize the total cost of false positives and threshold changes, subject to the constraint that any attack against these thresholds will result in at most $P$ damage. In order to solve this problem, we use a dynamic-programming algorithm. 
We will first discuss the algorithm without a cost for changing thresholds, and then show how to extend it to consider costly threshold changes.

For any two variables $n$ and $m$ such that $n \in \{1, \ldots, T\}$ and $0 \le m < n$, we define $\algoname{Cost}(n, m)$ to be the minimum cost of false positives from $n$ to $T$ subject to the damage constraint $P$, given that we only have to consider attacks that start at $k_a \in \{n-m, \ldots ,T\}$ and that attacks are not detected prior to~$n$. If there are no thresholds that satisfy the damage constraint $P$ under these conditions, we let $\algoname{Cost}(n, m)$ be $\infty$.\footnote{Note that in practice, $\infty$ can be represented by a sufficiently high natural number.}

We can recursively compute $\algoname{Cost}(n, m)$ as follows. For any $n < T$ and $m$, iterate over all possible detection delay values $\delta_n$, and choose the one that results in the lowest cost $\algoname{Cost}(n, m)$. If $\delta_n > m$, then no attack could be detected at timestep $n$, and $\algoname{Cost}(n, m)$ would be the cost at timestep $n$ plus the minimum cost for timesteps $\{n+1, \ldots, T\}$ given that attacks may start at $\{n-m, \ldots,T\} = \{(n+1) - (m+1), \ldots,T\}$. On the other hand, if $\delta_n \le m$, then some attacks could be detected at timestep $n$, and the worst of these attacks would start at $n-m$. Hence, if $\sum_{k = n - m}^n \calD(k) \le P$, then $\algoname{Cost}(n, m)$ would be the cost at timestep $n$ plus the minimum cost for timesteps $\{n+1, \ldots,T\}$ given that attacks may start at $\{n+1-\delta_n, \ldots,T\}$. Otherwise, there would be an attack that could cause more than $P$ damage, so $\algoname{Cost}(n, m)$ would be $\infty$ by definition since there would be no feasible thresholds for the remaining timesteps. Formally, we let
\begin{equation}\label{Cost3}
\algoname{Cost}(n, m)= \min_{\delta_n}
\begin{cases}
\algoname{Cost}(n+1, m+1) + \FP(\delta_n),& \text{if } \delta_n >  m\\
\algoname{Cost}(n+1, \delta_n) + \FP(\delta_n) ,   & \text{else if } \displaystyle\sum_{k=n-m}^{n} \calD(k) \le P \\
\infty & \text{otherwise }
\end{cases} \;.
\end{equation}

Note that in the equation above, $\algoname{Cost}(n,m)$ does not depend on $\delta_1, \ldots, \delta_{n-1}$, it depends only on the feasible thresholds for the subsequent timesteps. Therefore, starting from the last timestep $T$ and iterating backwards, we are able to compute $\algoname{Cost}(n,m)$ for all timesteps $n$ and all values $m$. Note that for $n = T$ and any $\delta_T$, computing $\algoname{Cost}(T,m)$ is straightforward: if $\sum_{T-m}^T \calD(k) \le P$, then $\algoname{Cost}(T,m)=\FP(\delta_T)$; otherwise, $\algoname{Cost}(T,m) = \infty$. 

Having found $\algoname{Cost}(n, m)$ for all $n$ and $m$, $\algoname{Cost}(1,0)$ is by definition the minimum cost of false positives subject to the damage constraint $P$. The minimizing threshold values can be recovered by iterating forwards from $n = 1$ to $T$ and again using Equation~\eqref{Cost3}. That is, for every $n$, we select the threshold corresponding to the delay value $\delta^*_n$ that attains the minimum cost $\algoname{Cost}(n,m)$, where $m$ can easily be computed from the preceding delay values $\delta^*_1, \ldots, \delta^*_n$.\footnote{Note that in Algorithm~\ref{algo:adaptive}, we store the minimizing values $\delta^*(n, m)$ for every $n$ and $m$ when iterating backwards, thereby decreasing running time and simplifying the presentation of our algorithm.}

\paragraph{Costly Threshold Changes} Now, we show how to extend the computation of $\algoname{Cost}$ to consider the cost $C_d$ of changing the threshold. Let $\algoname{Cost}(n, m, \delta_{n-1})$ be the minimum cost for timesteps starting from $n$ subject to the same constraints as before but also given that the detection delay at timestep $n - 1$ is $\delta_{n-1}$. 
Then, $\algoname{Cost}(n, m, \delta_{n-1})$ can be computed similarly to $\algoname{Cost}(n, m)$: for any $n < T$, iterate over all possible detection delay values $\delta_n$, and choose the one that results in the lowest cost $\algoname{Cost}(n, m, \delta_{n-1})$. If $\delta_{n-1} = \delta_n$ or $n = 1$, then the cost would be computed the same way as in the previous case (i.e., similarly to Equation~\eqref{Cost3}). Otherwise, the cost would have to also include the cost $C_d$ of changing the threshold. Consequently, similarly to Equation \eqref{Cost3}, we define
\begin{equation}\label{Cost5}
\widehat{\algoname{Cost}}(n, m,\delta_{n-1})=
\begin{cases}
\algoname{Cost}(n+1, m+1,\delta_n) + \FP(\delta_n)& \text{if } \delta_n >  m\\
\algoname{Cost}(n+1, \delta_n,\delta_n) + \FP(\delta_n)   & \text{if } \displaystyle\sum_{k=n-m}^{n} \calD(k) \le P \\
\infty & \text{otherwise }
\end{cases},
\end{equation}
and then based on the value of $\delta_{n-1}$, we can compute $\algoname{Cost}(n, m, \delta_{n-1})$ as
\begin{equation}\label{Cost4}
\algoname{Cost}(n, m, \delta_{n-1})= \min_{\delta_n}
\begin{cases}
\widehat{\algoname{Cost}}(n,m,\delta_{n-1}) & \text{if } \delta_n = \delta_{n-1} \lor n=1\\
\widehat{\algoname{Cost}}(n,m,\delta_{n-1}) + C_d   & \text{otherwise }
\end{cases} \;.
\end{equation}
Note that for $n = 1$, we do not add the cost $C_d$ of changing the threshold. Similarly to the previous case, $\algoname{Cost}(1, 0, \text{arbitrary})$ is the minimum cost subject to the damage constraint $P$, and the minimizing thresholds can be recovered by iterating forwards.

\subsubsection{2) Optimal Damage Constraint}
For any damage value $P$, using the above dynamic-programming algorithm, we can find thresholds that minimize the total cost $TC(P)$ of false positives and threshold changes subject to the constraint that an attack can do at most $P$ damage. Since the defender's loss is the sum of its total cost and the damage resulting from a best-response attack, we can find optimal adaptive thresholds by solving
\begin{equation}
\min_P \; TC(P) + P 
\end{equation}
and computing the optimal thresholds $\veta^*$ for the minimizing $P^*$ using our dynamic-programming algorithm.

To show that this formulation does indeed solve the problem of finding optimal adaptive thresholds, we use indirect proof.
For the sake of contradiction, suppose that there exist thresholds $\veta'$ for which the defender's loss $\calL'$ is lower than the loss $\calL^*$ for the solution $\veta^*$ of the above formulation.
Let $P'$ be the damage resulting from the attacker's best response against $\veta'$, and let $TC'$ be the defender's total cost for $\veta'$.
Since the worst-case attack against $\veta'$ achieves at most $P'$ damage, we have from the definition of $TC(P)$ that $TC' \geq TC(P')$.
It also follows from the definition of $TC(P)$ that $L^* \leq TC(P^*) + P^*$.
Combining the above with our supposition $L^* > L'$, we get
$$TC(P^*) + P^* \geq L^* > L' = TC' + P' \geq TC(P') + P' .$$
However, this is a contradiction since $P^*$ minimizes $TC(P) + P$ by definition.
Therefore, $\veta^*$ must be optimal.

It remains to show that Algorithm~\ref{algo:adaptive} finds an optimal damage value $P^*$.
To this end, we show that $P^*$ can be found in polynomial time using an exhaustive search.
Consider the set of damage values $\bar{\calD}(k_a,k_e)$ from all possible attacks $k_a \leq k_e$, that is, the set $$\left\{ \sum_{k = k_a}^{k_e} \calD(k) \, \middle|\, k_a \in \{1, \ldots, T\}, k_e \in \{k_a, \ldots, T\} \right\} .$$
Let the elements of this set be denoted by $P_1, P_2, \ldots$ in increasing order. It is easy to see that for any $i$, the set of thresholds that satisfy the constraint is the same for every $P \in [P_i, P_{i+1})$. Consequently, for any $i$, the cost $TC(P)$ is the same for every $P \in [P_i, P_{i+1})$. Therefore, the optimal $P^*$ must be a damage value $P_i$ from the above set, which we can find by simply iterating over the set.
\qed
\end{proof}

\begin{proposition}
The running time of Algorithm~\ref{algo:adaptive} is $\calO(T^4 \cdot|\Delta|^2)$.
\end{proposition}

Note that since possible detection delay values can be upper-bounded by $T$, the running time of Algorithm~\ref{algo:adaptive} is also $\calO(T^6)$.

\begin{proof}
	In the dynamic-programming algorithm, we first compute $\algoname{Cost}(n,m,$ $\delta_{n-1})$ for every $n\in\{1,\ldots,T\}$, $m\in\{1,\ldots,n-1\}$, and $\delta_{n-1} \in \Delta$, and each computation takes $\calO(|\Delta|)$ time. Then, we recover the optimal detection delay for all timesteps $\{1,\ldots,T\}$, and the computation for each timestep takes a constant amount of time. Consequently, the running time of the dynamic-programming algorithm is $\calO(T^2 \cdot |\Delta|^2)$.
	 
In the exhaustive search, we first enumerate all possible damage values by iterating over all possible attacks $(k_a, k_e)$, where $k_a \in \{1, \ldots, T\}$ and $k_e \in \{k_a, \ldots, T\}$. Then, for each possible damage value, we execute the dynamic-programming algorithm, which takes $\calO(T^2 \cdot |\Delta|^2)$ time. Consequently, the running time of Algorithm~\ref{algo:adaptive} is $\calO(T^4 \cdot |\Delta|^2)$. \qed
\end{proof} 

Finally, note that the running time of the algorithm can be substantially reduced in practice by computing $\algoname{Cost}$ in a lazy manner: starting from $n=1$ and $m=0$, compute and store the value of each $\algoname{Cost}(n, m, \delta_{n-1})$ only when it is referenced, and then reuse it when it is referenced again. Unfortunately, this does not change the worst-case running time of the algorithm.

\section{Numerical Results}
\label{sec:numerical}
In this section, we evaluate our approach numerically using an example. In particular, we consider the anomaly-based detection of deception attacks in water distribution networks.
In such networks, an adversary may compromise pressure sensors deployed to monitor the leakages and bursts in water pipes. By compromising sensors, adversary may alter their true observations,
which can then result in physical damage and financial losses. Next, we present the system model and the simulations of our results.

\subsubsection{System Model.} Figure~\ref{fig:daydemand} presents hourly water demand for a water network during a day \cite{djurin2014analysis}.
Since demand is time-dependent, the expected physical damage and financial loss caused by an attack on sensors is also time-dependent. That is, the expected disruptions at a high-demand time would be more problematic than the disruptions at a low-demand time. Therefore, for each timestep $k \in ~\{1,...,24\}$, we can define the expected damage as $\calD(k)=\alpha\cdot d(k)$ where $d(k)$ is the demand at time $k$, and $\alpha\in\Real_+$ is a fixed value for scaling (for example, water price rate). In our experiments, we let $\alpha=2$.
\begin{figure}
	\centering
	\includegraphics[width=10.5cm]{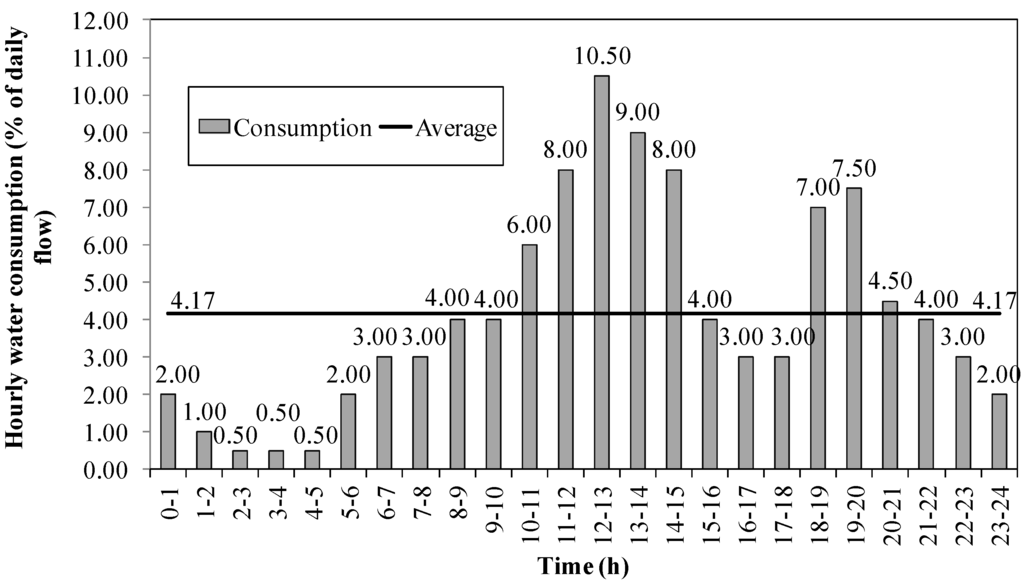}
	\caption{Hourly water demand during a day \cite{djurin2014analysis}.}
	\label{fig:daydemand}
\end{figure}

To discover attacks, we use anomaly-based detection systems implementing sequential change detection. 
Based on the results presented in \cite{cardenas2011attacks}, we derive the attainable detection delays and false alarm rates for the detector as shown in Figure~\ref{fig:tradeoff}. We observe that for the detection delay $\delta=0$, the false positive rate is $\FP(\delta)=0.95$, and for $\delta=23$, the false positive rate is $\FP(\delta)=0.02$. As expected, the detection delay is proportional to the threshold, and the false positive rate is inversely proportional to the threshold~\cite{basseville1993detection}.

\begin{figure}[h!]
	\centering
	\includegraphics[width=9cm]{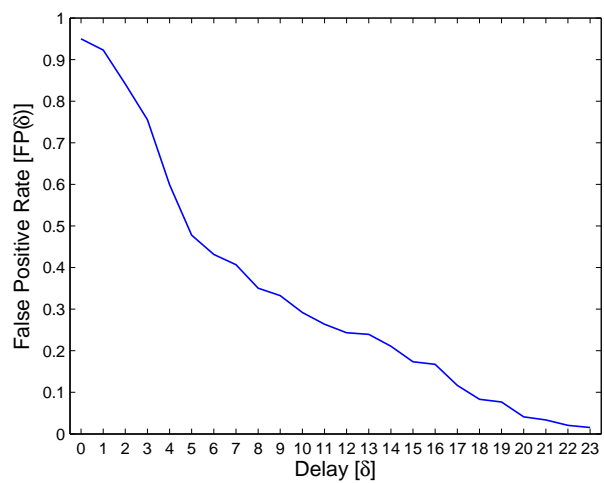}
	\caption{Trade-off between the detection delay and the false positive rate.}
	\label{fig:tradeoff}
\end{figure}

\subsubsection{Fixed Threshold.}
In the case of fixed threshold, the objective is to select the strategy that minimizes the defender's loss~\eqref{eq:defenderLoss} while assuming the attacker will respond using a best-response attack. Letting $C=7$ and using Algorithm~\ref{algo:fixed}, we obtain $\delta^\ast=5$, and the optimal loss $L^* = 171.30$. Figure~\ref{fig:fixed} shows the best-response attack corresponding to this threshold value. The best-response attack starts at $k_a^*=10$ and attains the payoff $P^*=\sum_{k=10}^{15} \calD(k)=91$. Note that if the attacker starts the attack at any other timestep, the damage caused before detection is less than $P^*$. 
\begin{figure}[h!]
	\centering
	\includegraphics[width=9cm]{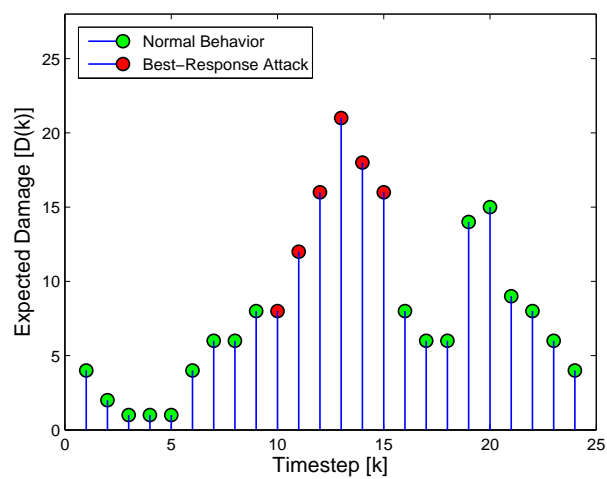}
	\caption{Best-response attack corresponding to the optimal fixed threshold $\delta^\ast=5$.}
	\label{fig:fixed}
\end{figure}

Next, letting $C=8$, we obtain $\delta^\ast=6$ as the optimal defense strategy, which leads to the optimal loss $L^\ast=181.86$, and best-response attack $k_a^\ast=9$, with the payoff $P^\ast = 99$. We observe that, as expected, the optimal delay is higher for the case of false alarms with higher costs.

\subsubsection{Adaptive Threshold.}
Using the same setting, we use Algorithm~\ref{algo:adaptive} to find an optimal adaptive threshold. We let $C=8$ and $C_d=10$. As shown in Figure~\ref{fig:adaptive}, we obtain the optimal adaptive threshold $\delta(k)=~23$ for $k\in\{1,..,11\}$, $\delta(k)=~1$ for $\{12,..,15\}$, and $\delta(k)=3$ for $\{17,...,23\}$. The resulting optimal loss is $L^*=~138.88$\,. Figure~\ref{fig:adaptive} shows the corresponding best-response attack, which starts at $k_a=13$ and, attains the payoff $P^\ast=39$. This figure demonstrates that the detection threshold decreases as the system experiences high-demand, so that the attacks can be detected early enough. On the other hand, as the system experiences low-demand, the threshold increases to have fewer false alarms.
\begin{figure}[h!]
	\centering
	\includegraphics[width=9cm]{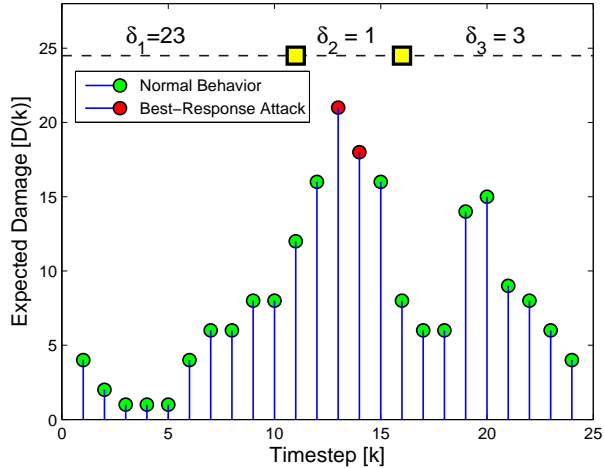}
	\caption{Best-response attack corresponding to the optimal adaptive threshold. The yellow points indicate the times at which the threshold change occurs.}
	\label{fig:adaptive}
\end{figure}

\subsubsection{Comparison.} Keeping $C=8$ fixed, Figure~\ref{fig:lossCd} shows the optimal loss as a function of cost of threshold change $C_d$. For small values of $C_d$, the optimal losses obtained by the adaptive threshold strategy are significantly lower than the loss obtained by the fixed threshold strategy. As the cost of threshold change $C_d$ increases, the solutions of adaptive and fixed threshold problems become more similar. In the current setting, the adaptive threshold solution converges to a fixed threshold when $C_d \geq 45$.
\begin{figure}[h!]
	\centering
	\includegraphics[width=9cm]{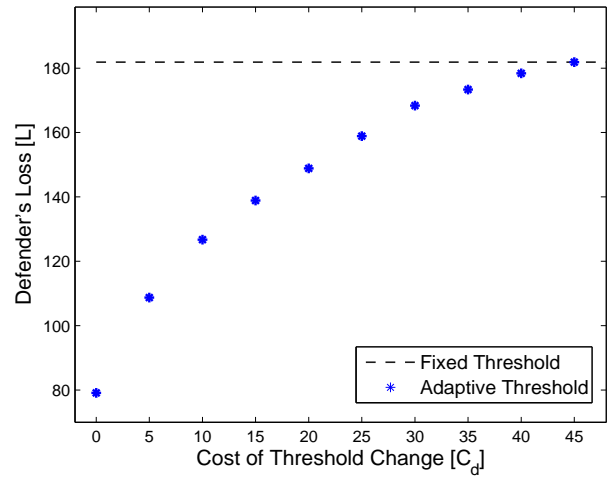}
	\caption{The defender's loss as a function of cost of threshold change.}
	\label{fig:lossCd}
\end{figure}

Furthermore, letting $C_d=8$, Figure~\ref{fig:lossC} shows optimal loss as a function of cost of false positives for fixed and adaptive threshold strategies. It can be seen that in both cases, the optimal loss increases as the cost of false alarms increases. However, in the case of adaptive threshold, the change in loss is relatively smaller than the fixed threshold.
\begin{figure}[h!]
	\centering
	\includegraphics[width=9cm]{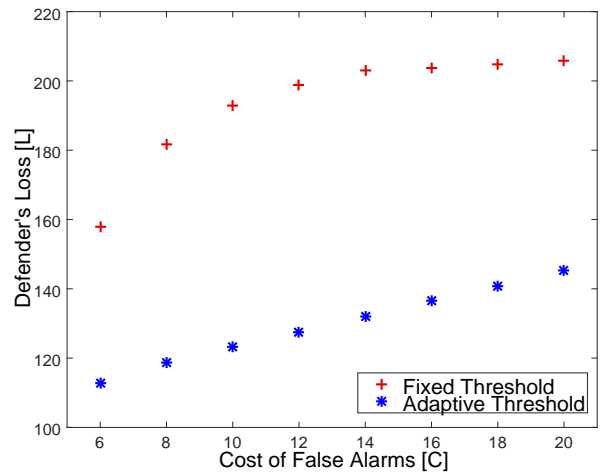}
	\caption{The defender's loss as a function of cost of false alarms.}
	\label{fig:lossC}
\end{figure}

\section{Related Work}

The problem of threshold selection for anomaly detection systems has been widely studied in the literature. Nevertheless, prior work has not particularly addressed the optimal threshold selection problem in the face of strategic attacks when the damage corresponding to an attack depends on time-varying properties of the underlying physical system.

Laszka et al. study the problem of finding detection thresholds for multiple detectors while considering time-invariant damages \cite{laszka2016optimal}. They show that the problem of finding optimal attacks and defenses is computationally expensive, thereby, proposing polynomial-time heuristic algorithms for computing approximately optimal strategies. Cardenas et al. study the use of physical models for anomaly detection, and describe the trade-off between false alarm rates and the delay for detecting attacks \cite{cardenas2011attacks}. Pasqualetti et al. characterize detection limitations for CPS and prove that an attack is undetectable if the measurements due to the attack coincide with the measurements due to some nominal operating condition \cite{pasqualetti2013attack}. 

Signaling games are also used to model intrusion detection \cite{casey2014cyber,estiri2010theoretical}. Shen et al. propose an intrusion detection game based on the signaling game in order to select the optimal detection strategy that lowers resource consumption \cite{shen2011signaling}. Further, Alpcan and Basar study distributed intrusion detection as a game between an IDS and an attacker, using a model that represents the flow of information from the attacker to the IDS through a network \cite{alpcan2003game,alpcan2004game}. The authors investigate the existence of a unique Nash equilibrium and best-response strategies. 

This work is also related to the FlipIt literature \cite{van2013flipit,laszka2013mitigating,laszka2014flipthem}. FlipIt is an attacker-defender game that studies the problem of stealthy takeover of control over a critical resource, in which the players receive benefits proportional to the total time that they control the resource. In \cite{pawlick2015flip}, the authors present a framework for the interaction between an attacker, defender, and a cloud-connected device. They describe the interactions using a combination of the FlipIt game and a signaling game.

In the detection theory literature, Tantawy presents a comprehensive discussion on design concerns and different optimality criteria used in model-based detection problems~\cite{tantawy2011model}. Alippi et al. propose a model of adaptive change detection that can be configured at run-time \cite{alippi2006adaptive}. This is followed by \cite{verdier2008adaptive}, in which the authors present a procedure for obtaining adaptive thresholds in change detection problems. 



\section{Concluding Remarks}

In this paper, we studied the problem of finding optimal detection thresholds for anomaly-based detectors implemented in dynamical systems in the face of strategic attacks. We formulated the problem as an attacker-defender security game that determines thresholds for the detector to achieve an optimal trade-off between the detection delay and the false positive rates.
To this end, first we presented an algorithm that computes optimal fixed threshold that is independent of time. Next, 
we defined adaptive threshold, in which the defender is allowed to change the detector's threshold with time. We provided a polynomial time algorithm to compute optimal adaptive threshold.
Finally, we evaluated our results using a case study. Our simulations indicated that the adaptive threshold strategy achieves a better overall detection delay-false positive trade-off, and consequently minimize the defender's losses, especially when the damage incurred by the successful attack varied with time.

In future work, we aim to extend this work by considering:
1) Multiple systems with different time-varying damage for each subsystem; 2) Sequential hypothesis testing detectors, in which there exits a trade-off between false alarm rate, missed detection rate, and detection delay; and 3) Moving target defense techniques based on randomized thresholds.

\subsection*{Acknowledgment}
This work is supported in part by the the National Science Foundation (CNS-1238959), Air Force Research Laboratory (FA 8750-14-2-0180), National Institute of Standards and Technology (70NANB15H263), Office of Naval Research (N00014-15-1-2621), and by Army Research Office (W911NF-16-1-0069).

\nocite{*}
\bibliographystyle{abbrv}
\bibliography{reference}

\end{document}